\documentclass[11pt]{article}
\usepackage{complexity}

\usepackage[english]{babel}
\usepackage[utf8x]{inputenc}
\usepackage[T1]{fontenc}
\usepackage{enumitem}
\usepackage[margin=1.in,marginparwidth=1.75cm]{geometry}

\usepackage[algo2e, ruled, vlined]{algorithm2e}
\usepackage{amsmath}
\usepackage{amsfonts}
\usepackage{graphicx}
\usepackage{amsthm}
\usepackage{amssymb}
\usepackage{dsfont}
\usepackage[algo2e, ruled, vlined]{algorithm2e}
\usepackage[noend]{algorithmic} 
\usepackage{bbm}

\usepackage[T1]{fontenc}
\usepackage{microtype}          %
\usepackage[full]{textcomp}

\usepackage{lipsum}
\usepackage{xspace}
\usepackage{relsize}
\usepackage{framed}
\usepackage{xcolor}
\usepackage{graphicx}
\usepackage{multirow}
\usepackage{subcaption}
\usepackage{float}

\usepackage[colorinlistoftodos]{todonotes}
\usepackage[colorlinks=true, allcolors=blue]{hyperref}

\newtheorem{theorem}{Theorem}
\newtheorem{proposition}{Proposition}

\newtheorem{claim}{Claim}
\newtheorem{lemma}{Lemma}

\newtheorem{remark}{Remark}

\theoremstyle{definition}

\newcommand{\knote}[1]{{\bf{\color{blue} [Karthik: #1]}}}
\newcommand{\cnote}[1]{{\bf{\color{magenta} [Chandra: #1]}}}

\usepackage{float}
\usepackage{booktabs}
\usepackage{multirow}
\usepackage{thmtools} 
\usepackage{thm-restate}
\usepackage{lipsum}

\newcommand{\Z}{\mathbb{Z}}
\renewcommand{\R}{\mathbb{R}}

\newif\ifdraft
\draftfalse

\newcommand{\distg}{\mathrm{dist}_{(G,\ell_G)}}
\newcommand{\disth}{\mathrm{dist}_{(H,\ell_H)}}
\newcommand{\Tminor}{\text{$T$-distance-monotone minor}\xspace}

\newcommand{\terminalminor}{\text{$T$-minor}\xspace}

\title{Improved Lower Bound for Steiner Point Removal}

\begin{document}
\author{Karthekeyan Chandrasekaran\thanks{Grainger College of Engineering, University of Illinois, Urbana-Champaign. Email: {\tt\{karthe, chekuri, weihaoz3\}@illinois.edu}. Supported in part by NSF grant CCF-2402667.}
\hspace{1.2em} Chandra Chekuri\footnotemark[1]
\hspace{1.2em} Qingyun Chen\thanks{University of California, Santa Cruz. Email: {\tt qingyun.chen152@gmail.com}. Supported in part by NSF grant CCF-2423106 and CCF-2535599.}
\hspace{1.2em} Weihao Zhu\footnotemark[1]
}
\date{}
\maketitle

\pagenumbering{gobble}
\begin{abstract}
    In the Steiner Point Removal problem, we are given a graph $G=(V,E)$ with an edge-length function $\ell_G: E\rightarrow \R_+$ and a subset $T\subseteq V$ of terminals. The goal is to find a minor $H=(T, E_H)$ of $G$ on vertex set $T$ such that the shortest path metric derived from $G$ on the edges of $H$ preserves the distance between every pair of terminals within a small multiplicative stretch. Filtser \cite{filtser2018steiner} proved that a stretch of $O(\log |T|)$ can be achieved (in polynomial time), while Chen and Tan \cite{chen2024lower} more recently proved a lower bound of $\Omega\left(\sqrt{\frac{\log |T|}{\log\log |T|}}\right)$ on the achievable stretch.  Their lower bound is via a simple construction involving low-degree high-girth graphs. The existence of such graphs is guaranteed through the existence of low-degree high-girth expanders. In this work, we improve the lower bound to $\Omega(\sqrt{\log |T|})$ using the same simple construction of Chen and Tan but with a more careful analysis that exploits the expansion property. 
\end{abstract}

\pagenumbering{arabic}

\section{Introduction}\label{section:intro}
The Steiner Point Removal problem is about the existence and computation of a minor of a graph that approximately preserves distances between a given set of terminals. We set up the terminology needed to define the problem formally. Let $G=(V,E)$ be a graph with positive edge lengths $\ell_G:E\rightarrow \R_+$. We denote the graph along with the edge lengths by the tuple $(G, \ell_G)$. For a pair of vertices $u, v\in V$, we define $$\distg(u,v):=\min\left\{\sum_{e\in P}\ell_G(e): P \text{ is a } u-v \text{ path in } G \right\},$$
i.e., $\distg(u,v)$ is the shortest $u-v$ path distance in $G$ with respect to the length function $\ell_G$. 
A graph $H$ is a \emph{minor} of graph $G$ if $H$ is obtained from $G$ by a sequence of operations involving edge deletion, edge contraction, and vertex deletion. 
Let $T\subseteq V$ be a subset of vertices of $G$, that we call as terminals. Non-terminals of $G$ are known as Steiner points. 
A \emph{\terminalminor} of $G$ is a minor $H=(T, E_H)$ of $G$ whose vertex set is $T$. 
For a positive edge-length function $\ell_G: E\rightarrow \R_+$ associated with $G$ and a \terminalminor $H=(T, E_H)$ of $G$, we define the induced edge-length function $\ell_H: E_H\rightarrow \R_+$ as follows: for each edge $e=\{t_u, t_v\}\in E_H$, let $\ell_H(e):=\distg(t_u, t_v)$. 
For a graph $G=(V, E)$ with terminals $T$ and edge-lengths $\ell_G:E\rightarrow \R_+$ and a \terminalminor $H$ of $G$, the \emph{$\ell_G$-stretch} of $H$ is defined to be 
\begin{equation}
\max\left\{\frac{\disth(t_1,t_2)}{\distg(t_1, t_2)}: t_1, t_2\in T \text{ with } t_1\neq t_2\right\}. 
\end{equation}

In the Steiner Point Removal problem (denoted SPR), we are given an undirected graph $G=(V, E)$ with positive edge lengths $\ell_G:E\rightarrow \R_+$ and a set $T\subseteq V$ of $k$ vertices known as terminals. We will denote the number of terminals by $k$ throughout this work.
The goal is to find a \terminalminor $H$ of $G$ with minimum $\ell_G$-stretch.

Every $T$-minor $H$ of $G$ is \emph{distance-monotone}, i.e., $\disth(t_1, t_2)\ge \distg(t_1, t_2)$ for every pair of terminals $t_1, t_2\in T$ (see Lemma \ref{lemma:min-edge-length-minor}). 
A minor preserves structural information about the graph (e.g., planarity), while asking for low-stretch amounts to preserving distances. Thus, the Steiner Point Removal problem aims for a low-distortion distance minor---we want to have low-distortion in distances while also guaranteeing that $H$ excludes the same fixed minors as $G$. Gupta \cite{Gup01} studied the Steiner Point Removal problem in trees---for a tree $G$ with given edge length function $\ell_G$, he showed that a polynomial-time algorithm to construct a $T$-minor with $\ell_G$-stretch at most $8$. He also showed an unconditional lower bound: there exists a tree $G$ with edge length function $\ell_G$ and a subset $T$ of terminals such that every $T$-minor of $G$ has $\ell_G$-stretch at least $4(1-o(1))$. Chan, Xia, Konjevod, and Richa \cite{CXKR06} improved the lower bound to $8(1-o(1))$ on trees, thereby matching Gupta's upper bound. 

It is natural to ask whether arbitrary graphs admit low-stretch terminal minors. 
Kamma, Krauthgamer, and Nguy\~{\^e}n \cite{KKN15} gave a polynomial-time algorithm that finds a \terminalminor with $\ell_G$-stretch $O(\log^5{|T|})$ for a given graph $(G, \ell_G)$. 
The algorithmic upper bound was improved to $O(\log^2{|T|})$ by Cheung \cite{Che18} and subsequently, to $O(\log{|T|})$ by Filtser \cite{filtser2018steiner}. Chang, Conroy, Le, Milenkovi\'{c}, Solomon, and Than showed that an $O(1)$-factor stretch is achievable in planar graphs and proper minor-closed families of graphs \cite{CCLMST24minorfree}. 
On the lower bound side, until recently, only the $8(1-o(1))$ bound for trees was known.
Chen and Tan \cite{chen2024lower} obtained a lower bound of $\Omega(\sqrt{\log{|T|}/\log{\log{|T|}}})$. 
We emphasize that the lower bound is unconditional and does not rely on any computational intractability assumption. We improve the lower bound in this work.

\begin{theorem}\label{theorem:main-theorem}
    For every positive integer $K$, there exists a graph with positive edge lengths $(G=(V,E), \ell_G:E\rightarrow \R_+)$ and a set $T\subseteq V$ of $k$ terminals with $k>K$, such that every \terminalminor of $G$ has $\ell_G$-stretch $\Omega(\sqrt{\log k})$.
\end{theorem}

Our lower bound instance is almost identical to that of Chen and Tan \cite{chen2024lower}. Their construction is  simple. Start from a low-degree high-girth graph $G$ and declare all edge lengths as unit.
For each vertex $v$ of $G$ add a pendant terminal $t_v$ with an edge of length of $L = \Theta(\sqrt{\log{k}\log{\log{k}}})$, and make all the original vertices as Steiner vertices.  The existence of low-degree high-girth graphs are guaranteed through Ramanujan graphs, which are also expanders. Our construction is essentially the same except that we use low-degree high-girth ``expander'' and change $L$ to $\Theta(\sqrt{\log{k}})$. We note that the lower bound construction of ``expanders with pendants/tails'' has appeared in the literature before in the context of integrality gaps for $0$-extension and metric labeling \cite{CKR05, KKMR09, CT24}.

Although the proof of Chen and Tan is short, its key step is a cleverly chosen quantity measuring the size of a collection of short paths needed to cover a fixed graph path. Our proof instead proceeds through simple and modular steps. Assuming low stretch, we show that the shortest path underlying each minor edge crosses exactly once between contraction clusters. Expansion and a direct counting argument then produce a short path containing many such crossings. Finally, large girth forces the corresponding terminal path either to have a long image in the original graph or to use many minor edges, in both cases contradicting small stretch.
In Section \ref{sec:tightness}, we show that our analysis is tight for the instances we construct. 
Thus, improving the lower bound further may require a substantially different construction.
\section{Preliminaries}\label{section:prelim}

Let $H$ be a \terminalminor of $G$. The following lemma shows that for every pair of terminals, the distance between the pair in $(H, \ell_H)$ is at least the distance between the pair in $(G, \ell_G)$.


\begin{lemma}\cite{chen2024lower}\label{lemma:min-edge-length-minor}
    Let $G=(V, E)$ be an undirected graph with positive edge lengths $\ell_G:E\rightarrow \R_+$ and $T\subseteq V$ be a subset of vertices. For every \terminalminor $H=(T, E_H)$ of $G$, let $\ell_H:E_H\rightarrow \R_+$ be the function defined as $\ell_H(e):=\distg(t_u, t_v)$ for every $e=\{t_u,t_v\}\in E_H$. Then, $\disth(t_1, t_2)\ge \distg(t_1, t_2)$ for every $t_1, t_2\in T$.
\end{lemma}
\begin{proof}
    For every two terminals $t_1, t_2 \in T$, let path $P=(p_0, p_1, \ldots, p_r)$ be the shortest path in $(H, \ell_H)$ connecting $t_1$ and $t_2$, where $p_0=t_1$ and $p_r=t_2$. Hence, we have
    $$\begin{aligned}
        \disth (t_1, t_2) &= \sum_{i=0}^{r-1} \disth (p_i, p_{i+1}) \\
        &= \sum_{i=0}^{r-1} \distg (p_i, p_{i+1}) \ \ \ \text{(by definition of $\ell_H$)} \\
        &\geq \distg (t_1, t_2).
    \end{aligned}$$
\end{proof}


For simplicity, all logarithms are by default to the base of $2$. Let $G=(V,E)$ be an undirected graph with positive edge lengths $\ell_G:E\rightarrow \R_+$. We will denote such a graph by the tuple $(G, \ell_G)$. If we do not specify the edge lengths, then all edge lengths are assumed to be unit.
For every path $P$ in graph $G$, its length is defined as the sum of the length of its edges, which is denoted by $\ell_G(P)$. We say two paths $P=(p_0, p_1, \ldots, p_{r_1})$ and $Q=(q_0, q_1, \ldots, q_{r_2})$ are different if $r_1\neq r_2$ or $p_i\neq q_i$ for some $0\leq i \leq r_1$. We say $P\subseteq Q$ if every edge in $P$ is contained in $Q$. We define $P \cup Q$ as the set of edges belonging to $P$ or $Q$. Let $\mathrm{girth}(G, \ell_G)$ be the length of the shortest cycle in $(G, \ell_G)$. For every positive integer $d\geq 1$, we say that graph $G$ is \textit{$d$-regular} if every vertex in $V$ has degree exactly $d$. For every subset $S\subseteq V$ of vertices, we define $\delta(S)$ as the set of edges intersecting both $S$ and $V\setminus S$ and $G[S]$ as the subgraph of $G$ induced by $S$.

\section{The Instance}\label{section:instance}

In this section, we will use low-degree high-girth expander (i.e., Ramanujan graphs) to construct our large-stretch instance.

\begin{lemma}\label{lemma:cayley-graph}\cite{lubotzky1988ramanujan,alon2021high}
    For every positive integer $K$, there exists an undirected graph $G=(V,E)$ with $k>K$ vertices that has the following properties: 
    \begin{enumerate}
        \item $G$ is $3$-regular,
        \item $|\delta(S)|\geq \frac{|S|}{20}$ for every $S\subseteq V$ with $|S|\leq \frac{k}{2}$, and 
        \item $\mathrm{girth}(G)\geq \frac{\log k}{10}$.
    \end{enumerate}
\end{lemma}

For a given positive integer $K$, let $G_0=(V_0,E_0)$ be the graph guaranteed to exist in Lemma~\ref{lemma:cayley-graph}. For every vertex $u\in V_0$, we add a new vertex $t_u$ that is adjacent to $u$. Let $G=(V,E)$ be the resulting graph. Let $L:=\lceil \sqrt{\log{k}} \rceil$. For every edge $e\in E$, we set $\ell_G(e):=1$ if $e\in E_0$, and $\ell_G(e):=L$ otherwise. We note that $G$ has $2k$ vertices. We define the set $T:=\{t_u:u\in V_0\}$ as the set of terminals. Figure~\ref{figure:instance} illustrates an instance when $G_0$ contains $4$ vertices. 

\begin{figure}[h]
    \centering
    \includegraphics[width=0.6\linewidth]{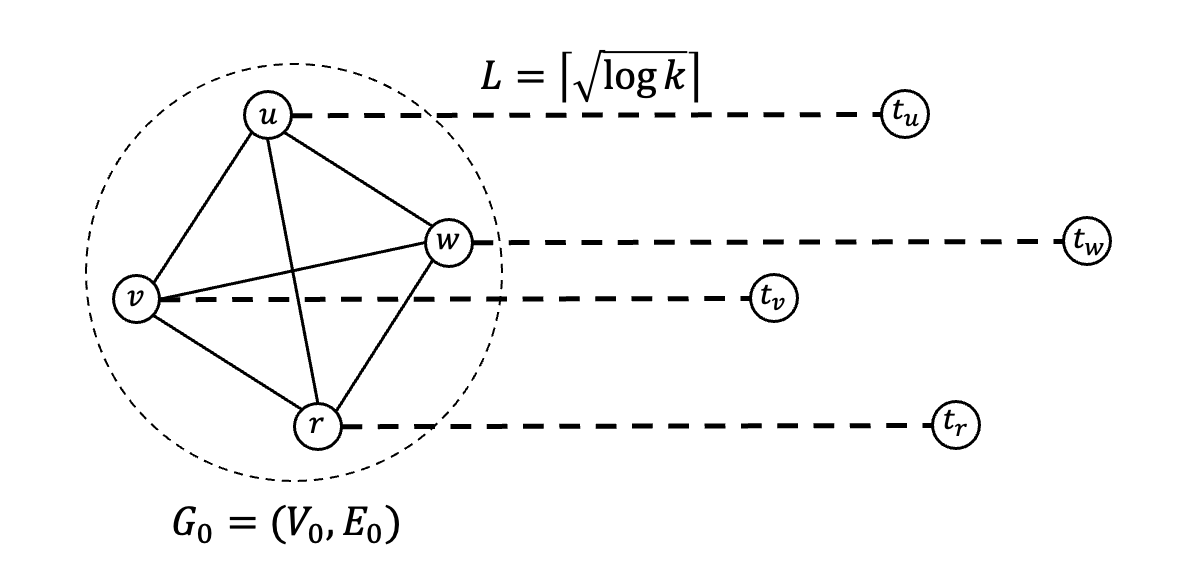}
    \caption{The figure shows graph $G=(V,E)$ when $G_0$ contains $k=4$ vertices, where $V_0=\{u,v,w,r\}$ are the four vertices on the left side and their corresponding terminals are on the right side. Every edge $e\in E_0$ has length $\ell_G(e)=1$, while every other edge $e\in E\setminus E_0$ has length $\ell_G(e)=L=\lceil \sqrt{\log k} \rceil$.
    }
    \label{figure:instance}
\end{figure}

\section{Properties of the instance}

We now summarize several useful properties of the instance. First, a \terminalminor $H$ of $G$ can be obtained by only deleting edges and contracting edges and without any vertex deletions, because if a non-terminal vertex $u$ is deleted, then the terminal $t_u$ in $H$ will be isolated from any other vertex. Consequently, the \terminalminor $H$ can be represented by a mapping $f:V\rightarrow T$ from vertices to terminals, where for every non-terminal $u$, $f(u)$ denotes the terminal to which the non-terminal vertex $u$ is contracted. For every terminal $t_u \in T$, we define $f(t_u):=t_u$. For every vertex $u\in V_0$, we define $B_u:=\{v\in V_0: f(v)=t_u\}$. We observe that $G_0[B_u]$ is a connected subgraph for every $u\in V_0$ if $B_u$ is non-empty.

The following proposition shows that every edge in a \terminalminor of $G$ has edge-length at least $L$. 
\begin{proposition}\label{proposition:edge-length}
    Let $H=(T, E_H)$ be a \terminalminor of $G$. Then, every edge $e=\{t_u, t_v\}\in E_H$ has $\ell_H(e)\ge L$.
\end{proposition}
\begin{proof}
    For every edge $e=\{t_u, t_v\}\in E_H$, we have
    $$\begin{aligned}
        \ell_H(e) = \distg(t_u, t_v) = \distg(u,v)+2L > L. 
    \end{aligned}$$
\end{proof}

The next proposition shows a girth property. 
\begin{proposition}\label{proposition:girth}
    For every pair of distinct vertices $u,v \in V_0$, let $P$ be a $u-v$ path in $G$ and $Q$ be a $u-v$ walk in $G$. Then, either $P\subseteq Q$ or $\ell_G(P)+\ell_G(Q)\geq \mathrm{girth}(G, \ell_G)$.
\end{proposition}
\begin{proof}
    Assume that $P\not\subseteq Q$. Then, edges in $P \cup Q$ contains a cycle, which implies that $\ell_G(P)+\ell_G(Q)\geq \mathrm{girth(G, \ell_G)}$.
\end{proof}



We first show that a low-stretch minor cannot contract a vertex far from its terminal.


\begin{lemma}\label{lemma:same-terminal-distance}
    Suppose $H$ is a \terminalminor of $G$ with $\ell_G$-stretch $\alpha \leq \frac{\sqrt{\log k}}{200}$.
    For every pair of vertices $u, v\in V_0$ with $f(v)=t_u$,
    $$\distg(u,v)\leq \frac{\log k}{50}.$$
\end{lemma}
\begin{proof}
    For the case of $u=v$, we have $\distg(u,v)=0$.
    We may now assume that $u\neq v$. Since $f(v)=t_u$, we have $v\in B_u$. We note that $f(u)=t_u$, since $u$ is the only neighbor of terminal $t_u$ in $G$, which shows that $u\in B_u$. Because $G_0[B_u]$ is a connected subgraph of $G_0$, there exists a neighbor $v'$ of $v$ in $G_0$ with $f(v')=f(v)=t_u$. If $v'=u$, then $u$ and $v$ are adjacent in graph $G_0$ and hence, $\distg(u,v)=1$. Otherwise, $t_u$ is the only neighbor of terminals $t_v$ and $t_{v'}$ in \terminalminor $H$. Therefore,
    $$\begin{aligned}
        \alpha\geq \frac{\disth(t_v, t_{v'})}{\distg(t_v, t_{v'})} &=\frac{\disth(t_v,t_u)+\disth(t_u, t_{v'})}{\distg(t_v, t_{v'})} \\
        &\geq \frac{\disth(t_v,t_u)}{\distg(t_v, t_{v'})} \geq \frac{\distg(t_v,t_u)}{\distg(t_v, t_{v'})} \ \ \text{(by Lemma~\ref{lemma:min-edge-length-minor})} \\
        &=\frac{\distg(t_u, t_v)}{2L+1} \ \ \text{(since $v\sim v'$ in graph $G_0$)} \\
        &=\frac{2L+\distg(u,v)}{2L+1}.
    \end{aligned}$$
    We recall that $\alpha\leq \frac{\sqrt{\log k}}{200}$ and $L=\lceil \sqrt{\log k}\rceil$. This shows that
    $$\begin{aligned}
        \distg(u,v) &\leq \alpha\cdot (2L+1)-2L \leq \frac{\sqrt{\log k}}{200}\cdot (2\sqrt{\log k}+3)-2\sqrt{\log k} \leq \frac{\log k}{50}.
    \end{aligned}$$
\end{proof}

Our next lemma shows that if a vertex $v$ is contracted to a terminal $t_u$, then the shortest path between $t_u$ and $v$ in $G$ is unique. Moreover, every vertex on that path is also contracted to $t_u$.

\begin{lemma}\label{lemma:same-terminal-shortest}
    Suppose $H$ is a \terminalminor of $G$ with $\ell_G$-stretch $\alpha \leq \frac{\sqrt{\log k}}{200}$.
    For every pair of vertices $u, v\in V_0$ with $f(v)=t_u$, the shortest path $P$ between $t_u$ and $v$ in $G$ is unique. Moreover, for every vertex $w$ on path $P$, $f(w)=t_u$. 
\end{lemma}
\begin{proof}
    We first prove that the shortest path between $u$ and $v$ is unique. By Lemma~\ref{lemma:same-terminal-distance}, $\distg(u,v)\leq \frac{\log k}{50}$. If there exists two different shortest $u-v$ paths $P_1$ and $P_2$, then edges in $P_1\Delta P_2$ form one or more disjoint cycles, which implies that $\ell_G(P_1)+\ell_G(P_2)\geq \mathrm{girth}(G, \ell_G)$. This contradicts $\ell_G(P_1)=\ell_G(P_2)\leq \frac{\log k}{50}$ and $\mathrm{girth}(G, \ell_G)>\frac{\log k}{10}$. This immediately implies that the shortest path between $t_u$ and $v$ in $G$ is unique, since vertex $u$ is the only neighbor of terminal $t_u$ in $G$.

    Let $P_0$ be the unique shortest path between $u$ and $v$ in $G_0$. We observe that concatenating $\{t_u,u\}$ and $P_0$ gives the shortest path $P$ between $t_u$ and $v$. Since $f(t_u)=t_u$, it suffices to show that $B_u$ contains all vertices on the path $P_0$.
    We define
    $$I_u:=\left\{v\in V_0:\distg(u,v)\leq \frac{\log k}{50}\right\}.$$
    We observe that $u\in I_u$ and $G_0[I_u]$ is a connected subgraph of $G_0$. According to Lemma~\ref{lemma:same-terminal-distance}, $v\in I_u$, which further implies that $G_0[I_u]$ contains all vertices on the path $P_0$. Moreover, $G_0[I_u]$ is acyclic, since the girth of $(G,\ell_G)$ is at least
    $$\mathrm{girth}(G, \ell_G)\geq \frac{\log k}{10} > 2\cdot \frac{\log k}{50}+1.$$
    Thus, $P_0$ is the unique path between $u$ and $v$ in the subgraph $G_0[I_u]$.

    According to Lemma~\ref{lemma:same-terminal-distance}, $B_u\subseteq I_u$. We recall that $G_0[B_u]$ is a connected subgraph of $G_0$. Since $P_0$ is the unique path between $u$ and $v$ in the subgraph $G_0[I_u]$, it follows that $B_u$ contains all vertices on the path $P_0$. Hence, for every vertex $w$ on path $P$, we have that $f(w)=t_u$.
\end{proof}

We call an edge $e=\{a, b\}\in E$ to be a \emph{special edge} if $f(a)\neq f(b)$. Lemmas~\ref{lemma:same-terminal-distance} and \ref{lemma:same-terminal-shortest} together imply the following lemma.

\begin{lemma}\label{lemma:special-edge}
    Suppose $H$ is a \terminalminor of $G$ with $\ell_G$-stretch $\alpha \leq \frac{\sqrt{\log k}}{200}$.
    For every edge $e=\{t_u, t_v\}\in E_H$, the shortest path $P$ between $t_u$ and $t_v$ in $G$ is unique. Moreover, it contains exactly one special edge.
\end{lemma}
\begin{proof}
    Since $t_u\sim t_v$ in minor $H$, there exists an edge $e'=\{a,b\}\in E$ with $f(a)=t_u$ and $f(b)=t_v$. Let $P_1$ be the shortest path between $t_u$ and $a$ in $G$. If $a=t_u$, then $P_1$ is a trivial path containing the single vertex $t_u$ with $\ell_G(P_1)=0$. Otherwise, $a\in V_0$. By Lemma~\ref{lemma:same-terminal-shortest}, the shortest path between $t_u$ and $a$ is unique, and every vertex $w$ on path $P_1$ has $f(w)=t_u$. According to Lemma~\ref{lemma:same-terminal-distance},
    \[
        \ell_G(P_1)=\distg(t_u,a)=\distg(t_u,u)+\distg(u,a)\leq \frac{\log k}{50}+L.
    \]
    We may similarly define $P_2$ as the unique shortest path between $b$ and $t_v$ in $G$, where every vertex $w$ on path $P_2$ has $f(w)=t_v$ and $\ell_G(P_2)\leq \frac{\log k}{50}+L$.

    We now define a $t_u-t_v$ walk $P$ in $G$, which is obtained by concatenating $P_1$, $e'=\{a,b\}$, and $P_2$. We observe that $P$ is actually a $t_u-t_v$ path since $P_1$ and $P_2$ are vertex-disjoint. The length of path $P$ satisfies
    $$\begin{aligned}
        \ell_G(P) &= \ell_G(P_1)+\ell_G(e')+\ell_G(P_2) \leq \left(\frac{\log k}{50}+L\right)+L+\left(\frac{\log k}{50}+L\right) < \frac{\mathrm{girth}(G, \ell_G)}{2}.
    \end{aligned}$$
    If there exists another $t_u-t_v$ path $P'$ with length $\ell_G(P')\leq \ell_G(P)$, edges in $P\Delta P'$ form one or more disjoint cycles, which implies that $\ell_G(P)+\ell_G(P')\geq \mathrm{girth}(G, \ell_G)$. This contradicts $\ell_G(P')\leq \ell_G(P)<\frac{\mathrm{girth}(G, \ell_G)}{2}$. Hence, $P$ is the unique shortest $t_u-t_v$ path in $G$. Moreover, $e'$ is the only special edge in $P$.
\end{proof}

The following lemma crucially uses the expansion property of $G_0$ in addition to $3$-regularity and large girth, which allows us to show that if every terminal has at most half of the vertices of $G_0$ contracted into it, then 
there exists an $L$-length path in $G_0$ with a constant fraction of special edges. 

\begin{lemma}\label{lem:special-path}
        Suppose $H$ be a \terminalminor of $G$ with $|B_u|\leq k/2$ for every $u\in V_0$. 
        Then, there exists a path $P$ in $G_0$ such that $P$ has exactly $L$ edges, among which at least $L/60$ edges are special.
    \end{lemma}

    \begin{proof}
    Our proof is based on a counting argument. 
    We will first show that the number of $L$-edge paths is small. Next, we will show that the number of times special edges appear in $L$-edge paths is large. Using these two facts and the pigeonhole principle, we will conclude that there exists a $L$-edge path with a large number of special edges.

    We define 
    \[
    \mathcal{J}:=\{(p_0, p_1, \ldots, p_L): p_i\neq p_{j} \text{ for every distinct } i,j\in [0,L] \text{ and } \{p_{i-1},p_{i}\}\in E_0\ \forall\ i\in [L]\},
    \]
    i.e., $\mathcal{J}$ is the set of ordered $L$-edge paths in $G_0$. We now count the number of paths in $\mathcal{J}$. 
    \begin{claim}\label{claim:number-of-paths}
        \[
        |\mathcal{J}|=2^{L-1}\cdot 3k.
        \]
    \end{claim}
    \begin{proof}
        Consider the following procedure: pick an arbitrary vertex $p_0\in V_0$, pick $p_1$ to be an arbitrary neighbor of $p_0$, and pick $p_i$ to be an arbitrary neighbor of $p_{i-1}$ that is distinct from $p_{i-2}$ for every $i=2, 3, \ldots, L$. Since $L<\mathrm{girth}(G_0)$, the ordered set $(p_0, p_1, \ldots, p_L)$ constructed by this procedure is a path (without duplicated vertices) and hence, is present in $\mathcal{J}$. Moreover, every ordered set $(p_0, p_1, \ldots, p_L)\in \mathcal{J}$ can be constructed by this procedure. Thus, $(p_0, p_1, \ldots, p_L)\in \mathcal{J}$ if and only if it can be constructed by the procedure. The number of ordered sets that can be constructed by the procedure is $kd(d-1)^{L-1}$, where $d$ is the degree of each vertex in $G_0$: we have $k$ choices for vertex $p_0$, $3$ choices for vertex $p_1$, and $2$ choices for vertex $p_i$ for every $i\in \{2, 3,\ldots, L\}$. Hence, $|\mathcal{J}|=2^{L-1}\cdot 3k$.
    \end{proof}

    Next, we show that the number of times special edges appear in ordered $L$-edge paths is large. 
    \begin{claim}
        \label{claim:number-of-special-edges-in-paths}
        \[
        \sum_{P\in \mathcal{J}}\left|\left\{e\in P: e\text{ is special}\right\}\right|
        \ge 2^L\cdot L\cdot \frac{k}{40}. 
        \]
    \end{claim}
    \begin{proof}

    We begin with a lower bound on the number of special edges. Let $u\in V_0$. We recall that $|B_u|\le k/2$. By the edge-expansion property in Lemma~\ref{lemma:cayley-graph}, we have
    \[
    \left|\left\{e=\{a,b\}\in E_0: e \text{ is special and } |\{a,b\}\cap B_u|=1\right\}\right|=|\delta(B_u)|\geq \frac{|B_u|}{20}. 
    \]
    Consequently,
    \begin{align}
    \left|\left\{e\in E_0: e \text{ is special}\right\}\right|
    &= \frac{1}{2}\sum_{u\in V_0}\left|\left\{e=\{a,b\}\in E_0: e \text{ is special and } |\{a,b\}\cap B_u|=1\right\}\right| \notag\\
    &\ge \frac{1}{2}\sum_{u\in V_0}\frac{|B_u|}{20} 
    =\frac{k}{40}. \label{eq:many-special-edges}
    \end{align}

    Next, we count the number of ordered $L$-edge paths in $G_0$ containing a fixed edge $e$. 
    Let $e=\{u,v\}\in E_0$ be an arbitrary edge. 
    Consider the following process: pick an arbitrary index $i \in [L]$ and let $\{p_{i-1}, p_i\}=\{u,v\}$. Next, for every $j=i-2, \ldots, 0$, pick $p_j$ to be an arbitrary neighbor of $p_{j+1}$ that is distinct from $p_{j+2}$ and for every $j=i+1, \ldots, L$, pick $p_j$ to be an arbitrary neighbor of $p_{j-1}$ that is distinct from $p_{j-2}$. 
    Since $L<\mathrm{girth}(G_0)$, the ordered set $(p_0, p_1, \ldots, p_{i-1}, p_{i}, \ldots, p_L)$ is a path (without duplicated vertices). Moreover, the $i$'th edge in this path is the edge $e$. We note that every ordered set $(p_0, p_1, \ldots, p_L)\in \mathcal{J}$ containing the edge $e$ can be constructed by this procedure. Hence, an ordered set $(p_0, p_1, \ldots, p_L)\in \mathcal{J}$ contains the edge $e$ if and only if it can be constructed by this procedure. 
    The number of ordered sets that can be constructed by this procedure is $2^L\cdot L$ (we have $L$ choices for the index $i$, two choices for $\{p_i,p_{i-1}\}=\{u,v\}$, and $p_j$ has $2$ choices for every $j\in\{0,1,2,\ldots,L\}\setminus \{i,i-1\}$). Hence, 
    \begin{equation}
        \left|\left\{(p_0,p_1,\ldots, p_L)\in \mathcal{J}: \exists i\in [L] \text{ with }\{p_{i-1},p_i\}=\{u,v\}\right\}\right|=2^L\cdot L. \label{eq:number-of-paths-with-a-fixed-edge}
    \end{equation}

    Hence,
    $$\begin{aligned}
        \sum_{P\in \mathcal{J}} \left|\left\{e\in P: e \text{ is special}\right\}\right|
        &= \sum_{P\in \mathcal{J}} \left(\sum_{e\in E_0: e \text{ is special}} 1_{e\in P}\right)= \sum_{e\in E_0: e \text{ is special}} \left(\sum_{P\in \mathcal{J}} 1_{e\in P}\right)\\
        &= \sum_{e\in E_0: e\text{ is special}} \left|\left\{(p_0,p_1,\ldots, p_L)\in \mathcal{J}: \exists i\in [L] \text{ with } \{p_{i-1},p_i\}=e\right\}\right|\\
        &= 2^L\cdot L\cdot \left|\left\{e\in E_0: e \text{ is special}\right\}\right| \quad \quad \text{(by equation \eqref{eq:number-of-paths-with-a-fixed-edge})}\\
        &\geq 2^L\cdot L\cdot \frac{k}{40}. \quad \quad \text{(by inequality \eqref{eq:many-special-edges})}.
    \end{aligned}$$
    \end{proof}

    We now complete the proof of Lemma \ref{lem:special-path}. By Claims \ref{claim:number-of-paths} and \ref{claim:number-of-special-edges-in-paths} and by the pigeonhole principle, there exists a path $P\in \mathcal{J}$ that contains at least 
    \[
    \frac{\sum_{P\in \mathcal{J}} \left|\left\{e\in P: e \text{ is special}\right\}\right|}{|\mathcal{J}|}\geq \frac{2^L\cdot L\cdot \frac{k}{40}}{2^{L-1}\cdot 3k}=\frac{L}{60}
    \]
    special edges. Thus, there exists a path $P$ in $G_0$ such that $P$ has exactly $L$ edges, among which at least $L/60$ edges are special. 
    \end{proof}

\section{Proof of Large Stretch}\label{section:proof}

Let $G=(V, E)$ be the graph with terminals $T$ and edge-length function $\ell_G: E\rightarrow \R_+$ constructed in Section \ref{section:instance}. We show that every \terminalminor of $G$ has a large $\ell_G$-stretch as follows, which directly implies Theorem~\ref{theorem:main-theorem}.

\begin{lemma}
    Every \terminalminor $H$ of $G$ has $\ell_G$-stretch larger than $\frac{\sqrt{\log k}}{200}$.
\end{lemma}
\begin{proof}
    Suppose that $H$ is a \terminalminor of $G$ with $\ell_G$-stretch $\alpha\leq \frac{\sqrt{\log k}}{200}$.
    We recall that the \terminalminor $H$ can be represented by a mapping $f: V(G_0)\rightarrow T$, where $f(u)$ denotes the terminal to which the non-terminal vertex $u$ is contracted. Moreover, for every $u\in V_0$, we have $B_u=\{v\in V_0: f(v)=t_u\}$, i.e., $B_u$ is the set of non-terminal vertices that are contracted into the terminal $t_u$ in the minor $H$. We recall that an edge $e=\{u, v\}\in E_0$ is special if $f(u)\neq f(v)$. 

    Suppose there exists a vertex $u\in V_0$ with $|B_u|>\frac{k}{2}$. We recall that graph $G_0$ is $3$-regular. This implies that for every $h\in \Z_{\geq 0}$, the number of vertices $v\in V_0$ with $\distg(u,v)=h$ is at most $3^h$. Consequently,
    \[
    \left|\left\{v\in V_0: \distg(u,v)\le \frac{\log k}{50}\right\}\right|
    \le \sum_{h=0}^{\lfloor\frac{\log k}{50}\rfloor}3^h < 3^{\frac{\log k}{50}+1}<k/2<|B_u|.
    \]
    Hence, there exists a vertex $v\in B_u$ with $\distg(u,v)>\frac{\log k}{50}$, which contradicts Lemma~\ref{lemma:same-terminal-distance}.

    We may henceforth assume that for every vertex $u\in V_0$, $|B_u|\leq \frac{k}{2}$. 
    Let $P$ be a path in $G_0$ such that $P$ has exactly $L$ edges, among which at least $L/60$ edges are special---such a path is guaranteed to exist by Lemma \ref{lem:special-path}. Let $u$ and $v$ be the end vertices of path $P$. We note that $\distg(u,v)\le L$. Consequently, 
    \begin{equation}
        \distg(t_u, t_v)=2L+\distg(u,v)\le 3L. \label{eq:low-distance-in-G}
    \end{equation}
    Next, we show that $\disth(t_u,t_v)\geq \frac{\log k}{60}$. 

    \begin{claim}\label{claim:large-distance-in-H}
    \[
    \disth(t_u,t_v)\geq \frac{\log k}{60}.
    \]
    \end{claim}
    \begin{proof}
Let $Q=(q_0=t_u, q_1, \ldots, q_r=t_v)$ be a shortest $t_u-t_v$ path in $(H, \ell_H)$. For every $i\in [0, r-1]$, let $R_i$ be a shortest $q_i-q_{i+1}$ path in $(G, \ell_G)$. Let $R$ be the $t_u-t_v$ walk in $G$ obtained by concatenating the paths $R_0, R_1, \ldots, R_{r-1}$. We consider two cases:

\noindent \underline{\textsl{Case (a):}} Suppose $P\not\subseteq R$. We observe that the first edge in $R$ is $\{t_u,u\}$ and the last edge is $\{v,t_v\}$. Removing these two edges gives a $u - v$ walk. Hence, by Proposition \ref{proposition:girth}, $\ell_G(R) + \ell_G(P)\ge \mathrm{girth}(G, \ell_G)$. Thus,
\begin{align*}
\disth(t_u, t_v) =\ell_H(Q) &= \sum_{i=0}^{r-1}\disth(q_i, q_{i+1}) \\
    &\geq \sum_{i=0}^{r-1}\distg(q_i, q_{i+1}) \ \ \ \text{(by Lemma \ref{lemma:min-edge-length-minor})}\\
    &= \ell_G(R) \geq \mathrm{girth}(G, \ell_G)-\ell_G(P)\\
    &\geq \frac{\log k}{10} - L \geq \frac{\log k}{60}.
\end{align*}

\noindent \underline{\textsl{Case (b):}} Suppose $P\subseteq R$. We recall that $P$ contains $L/60$ special edges, which implies that $R$ also contains $L/60$ special edges. According to Lemma~\ref{lemma:special-edge}, for every $i\in [0, r-1]$, $R_i$ contains exactly one special edge. This shows that $r\geq L/60$. Thus,
$$\begin{aligned}
    \disth(t_u, t_v) &=\ell_H(Q) = \sum_{i=0}^{r-1}\ell_H(\{q_i, q_{i+1}\}) \geq \sum_{i=0}^{r-1} L \ \ \text{(by Proposition~\ref{proposition:edge-length})} \\
    &= r\cdot L \geq L^2/60 \geq \frac{\log k}{60}.
\end{aligned}$$
    \end{proof}
    
    Therefore, the \terminalminor $H$ has $\ell_G$-stretch
    $$\alpha\geq \frac{\disth(t_u, t_v)}{\distg(t_u, t_v)}\geq \frac{\log k/60}{3L}\geq\frac{\log k}{180(\sqrt{\log k}+1)}\geq\frac{\sqrt{\log k}-1}{180},$$
    which contradicts $\alpha\leq \frac{\sqrt{\log k}}{200}$ for sufficiently large $k$.
\end{proof}

\begin{remark}
    Our lower bound instance is almost identical to that of Chen and Tan \cite{chen2024lower} but our analysis is stronger. In particular, we exploit the expansion property of $G_0$ (see Lemma \ref{lem:special-path}), which was not the case in the analysis of Chen and Tan. 
\end{remark}
\section{Tightness of the instance}\label{sec:tightness}
The following lemma shows that the $\Omega(\sqrt{\log k})$ bound is tight for the instance constructed in Section~\ref{section:instance}. 

\begin{lemma}
    For the graph $(G, \ell_G)$ constructed in Section \ref{section:instance}, there exists a \terminalminor $H$ of $G$ with $\ell_G$-stretch $O(\sqrt{\log k})$.
\end{lemma}
\begin{proof}
    For every vertex $u\in V_0$, we contract $u$ and $t_u$ together. This gives us a \terminalminor $H$ of $G$ where, for every two different vertices $u, v\in V_0$, terminals $t_u$ and $t_v$ are adjacent in $H$ if and only if $u$ and $v$ are adjacent in $G$. Meanwhile, every edge $e$ in $H$ has length exactly $\ell_H(e)=2L+1$.

    For every two different vertices $u, v\in V_0$, suppose that the shortest $u-v$ path in $G_0$ is $P=(p_0=u, p_1, \ldots, p_r=v)$ of length $r$. Then, the shortest $t_u-t_v$ path in $(H, \ell_H)$ is $(t_{p_0}, t_{p_1}, \ldots, t_{p_r})$ with length $(2L+1)\cdot r$. Hence, we have
    $$\distg(t_u, t_v)=2L+\distg(u,v)=2L+r,$$ and
    $$\disth(t_u, t_v)=(2L+1)\cdot r.$$

    Thus,
    $$1\leq \frac{\disth(t_u, t_v)}{\distg(t_u, t_v)}=\frac{(2L+1)\cdot r}{2L+r}<2L+1=O(\sqrt{\log k}),$$
    which implies that the \terminalminor $H$ of $G$ has $\ell_G$-stretch $O(\sqrt{\log k})$.
\end{proof}

\paragraph{AI Disclosure.} The authors used ChatGPT to assist with the writing of this manuscript. The authors assume full responsibility for all content. 

\bibliographystyle{abbrv}
\bibliography{references}

@InProceedings{CT24,
  author =	{Chen, Yu and Tan, Zihan},
  title =	{{Lower Bounds on 0-Extension with Steiner Nodes}},
  booktitle =	{51st International Colloquium on Automata, Languages, and Programming (ICALP 2024)},
  pages =	{47:1--47:18},
  year =	{2024},
}

@article{CKR05,
author = {Calinescu, Gruia and Karloff, Howard and Rabani, Yuval},
title = {Approximation Algorithms for the 0-Extension Problem},
journal = {SIAM Journal on Computing},
volume = {34},
number = {2},
pages = {358-372},
year = {2005},
}

@article{KKMR09,
author = {Karloff, Howard and Khot, Subhash and Mehta, Aranyak and Rabani, Yuval},
title = {On Earthmover Distance, Metric Labeling, and 0-Extension},
journal = {SIAM Journal on Computing},
volume = {39},
number = {2},
pages = {371--387},
year = {2009},
}

@inproceedings{Che18,
author = {Cheung, Yun Kuen},
title = {Steiner point removal: distant terminals don't (really) bother},
year = {2018},
booktitle = {Proceedings of the 29th Annual ACM-SIAM Symposium on Discrete Algorithms},
pages = {1353--1360},
numpages = {8},
series = {SODA}
}

@inproceedings{Gup01,
author = {Gupta, Anupam},
title = {Steiner points in tree metrics don't (really) help},
year = {2001},
booktitle = {Proceedings of the 12th Annual ACM-SIAM Symposium on Discrete Algorithms},
pages = {220--227},
series = {SODA}
}

@inproceedings{CXKR06,
author = {Chan, T.-H. and Xia, Donglin and Konjevod, Goran and Richa, Andrea},
title = {A tight lower bound for the steiner point removal problem on trees},
year = {2006},
booktitle = {Proceedings of the 9th International Conference on Approximation Algorithms for Combinatorial Optimization Problems, and 10th International Conference on Randomization and Computation},
pages = {70--81},
series = {APPROX-RANDOM}
}

@article{KKN15,
author = {Kamma, L. and Krauthgamer, R. and Nguy\~{\^e}n, H. L.},
title = {{Cutting Corners Cheaply, or How to Remove Steiner Points}},
journal = {SIAM Journal on Computing},
volume = {44},
number = {4},
pages = {975-995},
year = {2015},
}

@inproceedings{filtser2018steiner,
  title={{Steiner point removal with distortion $O(\log k)$}},
  author={Filtser, Arnold},
  booktitle={Proceedings of the 29th Annual ACM-SIAM Symposium on Discrete Algorithms},
  pages={1361--1373},
  year={2018},
  series = {SODA},
}

@inproceedings{chen2024lower,
  title={{An $\tilde{\Omega}(\sqrt{\log{|T|}})$ Lower Bound for Steiner Point Removal}},
  author={Chen, Yu and Tan, Zihan},
  booktitle={Proceedings of the 35th Annual ACM-SIAM Symposium on Discrete Algorithms},
  pages={694--698},
  year={2024},
  series = {SODA},
}

@article{lubotzky1988ramanujan,
  title={Ramanujan graphs},
  author={Lubotzky, Alexander and Phillips, Ralph and Sarnak, Peter},
  journal={Combinatorica},
  volume={8},
  number={3},
  pages={261--277},
  year={1988}
}

@inproceedings{CCLMST24minorfree,
  author    = {Chang, Hsien-Chih and Conroy, Jonathan and Le, Hung and
               Milenkovi{\'c}, Lazar and Solomon, Shay and Than, Cuong},
  title     = {Shortcut Partitions in Minor-Free Graphs: {S}teiner Point Removal,
               Distance Oracles, Tree Covers, and More},
  booktitle = {Proceedings of the 2024 Annual ACM-SIAM Symposium on Discrete
               Algorithms (SODA)},
  pages     = {5300--5331},
  year      = {2024},
}

@article{alon2021high,
  title={High-girth near-Ramanujan graphs with localized eigenvectors},
  author={Alon, Noga and Ganguly, Shirshendu and Srivastava, Nikhil},
  journal={Israel Journal of Mathematics},
  volume={246},
  number={1},
  pages={1--20},
  year={2021},
  publisher={Springer}
}

\newpage
\appendix

\end{document}